\documentclass[11pt,letterpaper]{article}
\usepackage{fullpage}
\usepackage{amsfonts}
\usepackage{graphicx}
\usepackage{amsmath}
\usepackage{amsthm}
\usepackage{times}
\usepackage{mathptmx}
\usepackage{url}
\usepackage{graphicx}
\usepackage{amssymb}
\usepackage{amsmath}
\usepackage{subfigure}
\usepackage{ifpdf}
\ifpdf\fi%\usepackage{amsthm}
\DeclareMathAlphabet{\mathcal}{OMS}{cmsy}{m}{n}

\newtheorem{theorem}{Theorem}

\newcommand{\eat}[1]{}

\begin{document}
\renewenvironment{proof}{\noindent{\bf Proof:}}{\hspace*{\fill}\rule{6pt}{6pt}\bigskip}

\title{External-Memory Multimaps}
\author{Elaine Angelino \thanks{To appear in ISAAC 2011. A preliminary version of this work appears as a Brief Announcement at SPAA 2011.}  \\[6pt]
SEAS,  Harvard University\\
elaine@eecs.harvard.edu 
\and 
Michael T. Goodrich \thanks{Supported in part by the National Science
Foundation under grants 0724806, 0713046, and 0847968, and by the
Office of Naval Research under MURI grant N00014-08-1-1015.}\\[6pt]
Department of Computer Science, %School of Info. and Computer Sciences, 
University of California, Irvine,\\
goodrich@acm.org
 \and 
 Michael Mitzenmacher\thanks{Supported in part by 
the National Science Foundation under grants 0915922 and 0964473.} \\[6pt]
SEAS, Harvard University\\
michaelm@eecs.harvard.edu 
 \and 
 Justin Thaler\thanks{Supported by the Department of Defense (DoD) through the National Defense Science \& Engineering Graduate Fellowship (NDSEG) Program, and partially by NSF grant CNS-0721491.} \\[6pt]
SEAS, Harvard University\\
jthaler@seas.harvard.edu 
}

\date{}
\maketitle 

\vspace{-0.7 cm}

\begin{abstract}
Many data structures support dictionaries, also known as maps or
associative arrays, which store and manage a set of key-value pairs.
A \emph{multimap} is generalization that allows multiple values to be
associated with the same key.  
For example, the inverted file data structure that is used
prevalently in the infrastructure supporting search engines is a
type of multimap, where words are used as keys and document pointers
are used as values.
We study the multimap abstract data
type and how it can be implemented efficiently online in % parallel and
external memory frameworks, with constant expected I/O %or
%parallel-round 
performance.  
The key technique used to
achieve our results is a combination of cuckoo hashing using buckets that hold
multiple items with a multiqueue implementation to
cope with varying numbers of values per key.
Our external-memory results are for the standard two-level memory
model.% and our parallel results are for the bulk-synchronous parallel
%(BSP) model.
\end{abstract}

% \category{F.2.2}{Analysis of Algorithms and Problem
% Complexity}{Nonnumerical Algorithms and Problems}
% \terms{Algorithms, Theory}
% \keywords{MultiMap, Chernoff bounds, cuckoo hashing, randomized algorithms.}

%\ifFull\else
%\bigskip
%\centerline{\textbf{Regular submission to SPAA 2011}}
%\clearpage
%\fi

\section{Introduction}

A \emph{multimap} is a simple abstract data type (ADT) that generalizes the 
map ADT to support key-value associations in a way that allows
multiple values to be associated with the same key.
Specifically, it is a dynamic container, $C$, of key-value
pairs, which we call \emph{items}, 
supporting (at least) the following operations:
\begin{itemize}
\item
insert$(k,v)$: insert the key-value pair, $(k,v)$. This operation
allows for there to be existing key-value pairs having the same
key $k$, but we assume w.l.o.g. that the particular key-value pair
$(k,v)$ is itself not already present in $C$.
\item
isMember$(k,v)$: return true if %and only if 
the key-value pair, $(k,v)$, is
present in $C$.
\item
remove$(k,v)$: remove the key-value pair, $(k,v)$, from $C$.
This operation
returns an error condition if $(k,v)$ is not currently in $C$.
\item
findAll$(k)$: return the set of all key-value pairs
in $C$ having key equal to $k$.
\item
removeAll$(k)$: remove from $C$ all key-value pairs
having key equal to $k$.
\item
count$(k)$: Return the number of values associated with key $k$.
\end{itemize}

Surprisingly, we are not familiar with any previous discussion of this
abstract data type in the theoretical algorithms and data structures
literature.  Nevertheless, abstract data types equivalent to the above
ADT, as well as multimap implementations, are included in the C++
Standard Template Library (STL)~\cite{ms-tstrc-95}, Guava--the Google Java
Collections 
Library\footnote{\url{http://code.google.com/p/google-collections/}}, 
and the Apache Commons Collection 3.2.1 
API\footnote{\url{http://commons.apache.org/collections/apidocs/index.html}}.  
Clearly, the existence of these implementations provides
empirical evidence for the usefulness of this abstract data type.

\subsection{Motivation}
\label{sec:motivation}
One of the primary motivations for studying the multimap ADT is that
associative data in the real world can exhibit significant
non-uniformities with respect to the relationships between keys and
values.  For example, many real-world data sets follow 
a power law with respect to data frequencies indexed by rank.
The classic description of this law is that in a corpus of
natural language documents, defined with respect to $n$ words,
the frequency, $f(j,n)$, of the word of rank $j$ is predicted to be
\[
f(j,n) = \frac{1}{{j^s}{H_{N,s}}} ,
\]
where $s$ is a parameter characterizing the distribution and
$H_{n,s}$ is the $n$th generalized harmonic number.
Thus, if we wished to construct a data structure that can be used
to retrieve all instances of any query word, $w$, in such a corpus, 
subject to insertions and deletions of documents, then we could
use a multimap, but would require one that could handle large skews
in the number of values per key.  
In this case, the multimap could be viewed
as providing a dynamic functionality for
a classic static data structure, known as an \emph{inverted file}
or \emph{inverted index} (e.g., see Knuth~\cite{k-ss-73}).
Given a collection, $\Gamma$, of documents, an inverted file is an
indexing strategy that allows one to list, for any word $w$,
all the places in $\Gamma$ where $w$ appears.

Another powerful motivation for studying multimaps is graphical
data \cite{bb-cdvk-08}.  A multimap can represent a graph: keys
correspond to nodes, values correspond to neighbors, findAll
operations list all neighbors of a node, and removeAll operations
delete a node from the graph.  The degree distribution of many
real-life graphs follow a power law, motivating efficient handling of
non-uniformity.  

As a more recent example, static multimaps were used for a geometric
hashing implementation on graphical processing units
in \cite{asasmoa-rtphg-09}.  In this setting, signatures are computed from
an image, and a signature can appear multiple times in an image.  The
signature is a key, and the values correspond to locations where the
signature can be found.  Geometric hashing allows one to find query
images within reference images.  Dynamic multimaps could allow for
changes in reference images to be handled dynamically without
recalculating the entire structure.

There are countless other possible scenarios where we expect 
multimaps can prove useful.  In many settings, one can indicate the
intensity of an event or object by a score.  Examples include the
apparent brightness of stars (measured by stellar magnitudes), the
intensity of earthquakes (measured on the Richter scale), and the
loudness of sounds (measured on the decibel scale).  Necessarily,
when data from such scoring frameworks is labelled as key-value pairs
where the numeric score is the key, some scores will have
disproportionally many associated values than others.  In fact, in
assigning numeric scores to observed phenomena, there is a natural
tendency for human observers to assign scores that depend
logarithmically on the stimuli. This perceptual pattern is so common
it is known as the
\emph{Weber--Fechner Law}~\cite{Dehaene2003145,h-vdiwf-24}.
Multimaps may prove particularly effective for such data sets.

\subsection{Previous Related Work}

Inverted files have standard applications in 
text indexing (e.g., see Knuth~\cite{k-ss-73}),
and are important data structures
for modern search engines (e.g., see Zobel and Moffat~\cite{zm-iftse-06}).
Typically, this is a static structure and the collection $\Gamma$ is
usually thought of as all the documents on the Internet.
Thus, an inverted file is a static multimap
that supports the findAll($w$) operation (typically with a cutoff for
the most relevant documents containing $w$).

Cutting and Pedersen~\cite{cp-odiim-90} describe an inverted file
implementation that uses B-trees for the indexing structure and
supports incremental and batched insertions, but it doesn't support
deletions efficiently.
More recently, Luk and Lam~\cite{ll-eimef-07} describe an in-memory
inverted file implementation based on hash tables with chaining, but
their method also does not support fast deletions.
Likewise,
Lester {\it et al.}~\cite{lmz-eoict-08,lzw-eoimc-06} 
and B\"{u}ttcher {\it et al.}~\cite{bcl-himgt-06}
describe out-of-core
inverted file implementations that support insertions only.
B\"{u}ttcher and Clarke~\cite{bc-itvqt-05},
on the other hand, 
consider the trade-offs for allowing for
both insertions and deletions in an inverted file, 
and Guo {\it et al.}~\cite{gcxw-eoli-07}
describe a solution for performing such operations
by using a type of B-tree.

Our work utilizes a variation on cuckoo hash tables.  We assume the reader
has some familiarity with such hash tables, as originally presented by Pagh
and Rodler \cite{pr-ch-04}.\footnote{A general description can be found on Wikipedia
at \url{http://en.wikipedia.org/wiki/Cuckoo_hashing}.}  
We describe the relevant
background in Section~\ref{sec:map}.  

Finally, recent work by Verbin and Zhang \cite{vz} shows that in the external memory model,
for any dynamic dictionary data structure with query cost $O(1)$, the expected amortized cost of updates 
must be at least 1. As explained below, this implies our data structure is optimal up to constant factors. 
 
\subsection{Our Results}

In this paper we describe efficient %parallel and 
external-memory
implementations of the multimap ADT.  %Our parallel algorithms are
%designed for a cloud computing environment~\cite{rpcnh-mnmmb-09},
%abstracted using the bulk synchronous parallel (BSP)
%model~\cite{gv-dbspa-94,v-bmpc-90}, which has also been used for
%abstractions of GPU programming~\cite{hzg-bbsgp-08} and multi-core
%architectures~\cite{v-bmmcc-08}.  
Our external-memory algorithms are
for the standard two-level I/O model, which captures the memory hierarchy of
modern computer architectures (e.g.,
see~\cite{av-iocsr-88,DBLP:reference/algo/Vitter08}). In this model, there is a cache of size $M$ connected to a disk of
unbounded size, and the cache and disk are divided into blocks,
where each block can store up to $B$ items. %;
%the cache holds $M/B$ blocks.
Any algorithm can only operate on cached data,
and algorithms must therefore make memory transfer operations,
which read a block from disk into cache or vice versa.
The cost of an algorithm is the number of I/Os required,
with all other operations considered free.
All of our time bounds hold even when $M=O(B)$,
and we therefore omit reference to $M$
throughout. %Sections~\ref{sec:map} and \ref{sec:multimap}.

We support an online implementation of the multimap
abstract data type, where each operation must 
completely finish executing (either in terms of its data structure updates or
query reporting) prior to our beginning execution of any subsequent
operations.
The bounds we achieve for the multimap ADT methods are shown in
Table~\ref{tbl:bounds}. All bounds are \emph{unamortized}. 

\begin{table}[hbt]
\begin{center}
\begin{tabular}{|c|c|}
\hline
\textbf{Method} & \textbf{I/O Performance} \\
\hline
\rule[-8pt]{0pt}{22pt} insert$(k,v)$ & $\bar O(1)$ \\
\hline
\rule[-8pt]{0pt}{22pt} isMember$(k,v)$ & $O(1)$ \\
\hline
\rule[-8pt]{0pt}{22pt} remove$(k,v)$ & $O(1)$ \\
\hline
\rule[-8pt]{0pt}{22pt} findAll$(k)$ & $O(1+n_k/B)$ \\
\hline
\rule[-8pt]{0pt}{22pt} removeAll$(k)$ & $O(1)$ \\
\hline
\rule[-8pt]{0pt}{22pt} count$(k)$ & $O(1)$ \\
\hline
\end{tabular}
\end{center}
\caption{\label{tbl:bounds} Performance bounds for our multimap
implementation. $\bar O(*)$ denotes an expected bound.
Also, we use
$B$ to denote the block size, $N$ to denote the number of key-value
pairs, 
% $M$ to denote the memory size of each processor,
% $n=N/B$ to denote the number of blocks needed for the $N$ key-value pairs,
% $m=M/B$ to denote the number of blocks stored by each processor,
and
$n_k$ to denote the number of key-value pairs with key equal to $k$. 
%Finally, we note that the bound for removeAll$(k)$ follows from the other
%bounds and an amortization argument that charges each item's removal to its
%corresponding insertion.
}
\end{table}

Our constructions are based on the combination of two 
external-memory data structures---external-memory cuckoo hash tables
and multiqueues---which may be of independent interest.
We show that external-memory cuckoo hashing supports
a cuckoo-type method for insertions that %can be implemented in 
%a way that allows us to prove that only an expected constant number
provably requires only an expected constant number
of I/Os. %to find a location where each new item can be placed.
We then show that this performance can be combined with
expected constant I/O complexity for multiqueues to design a multimap
implementation that has constant (unamortized) worst-case or expected
I/O performance for most methods.
Our methods imply that one can maintain an inverted file in
external memory so as to support a constant 
expected number of I/Os for
insertions and worst-case constant I/Os for look ups and
item removal.
%Finally, we discuss how our solutions translate into parallel methods
%in the BSP model.

\section{External-Memory Cuckoo Hashing}
\label{sec:map}
In this section, we describe external-memory
versions of cuckoo hash tables with multiple items per bucket.  
The implementation we describe in this section is for the map ADT,
where all key-value pairs are distinct. We show later in this paper
how this approach can be used in concert with multiqueues to support
multiple key-value pairs with the same key for the multimap ADT.

Cuckoo hash tables that can store multiple items per bucket have been
studied previously, having been introduced
in \cite{dw-badtpcs-07}.  Generally the analysis has been
limited to buckets of a constant size, $d$, where here size is
measured in terms of the number of items, which in this context is a
key-value pair in our collection, $C$.  For our external-memory cuckoo
hash table, each bucket can store $B$ items, where $B$ is a %given
parameter defining our block size and is not necessarily a constant.

Formally, let ${\cal T}=(T_0,T_1)$ be a cuckoo hash table such that
each $T_i$ consists of $\gamma n/2$ buckets, where each bucket stores a
block of size $B$, with $n=N/B$.  (In the original cuckoo hash table setting,
$B$ = 1.)  One setting of particular interest
is when $\gamma = 1 + \epsilon$ for some (small) $\epsilon > 0$, so
that space overhead of the hash table is only an $\epsilon$ factor
over the minimum possible.  The items in $\cal T$ are indexed by keys
and stored in one of two locations, $T_0[h_0(k)]$ or $T_1[h_1(k)]$,
where $h_0$ and $h_1$ are random hash functions.  (The assumption that
the hash functions are random can be done away with using suitable
realistic hash functions; see for
example \cite{dw-badtpcs-07} for a discussion,
or \cite{mv-wshfw-08} for an alternative model.)

It should be clarified that, in some settings, the use of a cuckoo
hash function may be unnecessary or even unwarranted.  Indeed, if $B >
c \log n$ for a suitable constant $c$ and $\gamma = 1+\epsilon$, we
can use simple hash tables, with just one choice for each item,
instead.  In this case, with Chernoff and union bounds one can show
that with high probability all buckets will fit all the items that
hash to it, since the expected number of items per bucket will then be
$B/(1+\epsilon)$, and $B$ is large enough for strong tail bounds to
hold.  Cuckoo hashing here allows us to avoid such ``wide block
assumptions'', giving a more general approach.  In practice, also,
across the full range of possible values for $B$ we expect cuckoo
hashing to be much more space efficient.  Whether this space savings
is important may depend on the setting.

The important feature of the cuckoo hashing implementation is the way
it may reallocate items in $\cal T$ during an insertion.  Standard
cuckoo hashing, with one item per bucket, immediately evicts the
previous (and only) item in a bucket when a new item is to be inserted
in an occupied bucket.  With multiple items per bucket, there is a
choice available.  We describe what is known in this setting, and how
we modify it for our use here.

Let $G$ be the \emph{cuckoo graph}, where each bucket in $\cal T$ is a
vertex and, for each item $x$ currently in $C$, we connect
$T_0[h_0(x)]$ and $T_1[h_1(x)]$ as a directed edge, with the edge
pointing toward the bucket it is not currently stored in.  Suppose we
wish to insert an item $x$ into bucket $X$ in $\cal T$.  If $X$
contains fewer than $B$ items, then we simply add $x$ to $X$.
Otherwise, we need to make room for the new item.  

One approach for doing an insertion is to use a breadth first search
on the cuckoo graph.  The results of Dietzfelbinger and Weidling show
that for sufficiently large {\em constant} $B$, the expected insertion
time is constant.  Specifically, when $\gamma = 1 + \epsilon$ and
$B \geq 16 \ln (1/\epsilon)$, the expected time to insert a new key is
$(1/\epsilon)^{O(\log \log (1/\epsilon))}$, which is a constant.
(This may require re-hashing all items in very rare cases when an
item cannot be placed; the expected time remains constant.) Notice
that if $B$ grows in a fashion that is $\Omega(1)$, then a breadth
first search approach does not naturally take constant expected time,
as even the time to look if items currently in the bucket can be moved
will take $\Omega(B)$ time.  (It might still take constant expected
time -- it may be that only a constant number of buckets need to be
inspected on average -- but it does not appear to follow
from \cite{dw-badtpcs-07}.)

For non-constant $B$, we can apply the following mechanism:
we can use our buckets to mimic having $B/c$ distinct subtables for
some large constant $c$, where the $i$th subtable uses the $ci/B$th
fraction of the bucket space, and each item is hashed into a specific
subtable.  For $B = O(n^\delta)$ for $\delta < 1$, each subtable will contain
close to its expected number of items with high probability.  Further,
by choosing $c$ suitably large one can ensure that each subtable is
within a $1+\epsilon$ factor of its total space while maintaining an
expected $\log(1/\epsilon)^{O(\log \log (1/\epsilon))}$ insertion
time.  Specifically, we have the following theorem:
\begin{theorem} \label{thm:cuckoo}
Suppose for a cuckoo hash table ${\cal T}$ the block size satisfies
$B=\Omega(1)$ and $B=O(n^\delta)$ for $\delta < 1$.  Let $0 < \epsilon \leq 0.1$ be
arbitrary, let $C$ be a collection of $N$ items, and let $T$ be a
table with at least $(1+\epsilon)N/B$ blocks.  Suppose further we have
$B/c$ subtables, with $c = 16 \ln(1/\epsilon)$ with each item hashed
to a subtable by a fully random hash function, and the hash functions
for each subtable are fully random hash functions.
Finally, suppose the items of $C$ have been stored in ${\cal T}$ by
an algorithm using the partitioning process described above and the cuckoo hashing
process.  Then the expected time for the insertion of a new item $x$ using a BFS is
$(1/\epsilon)^{O(\log \log (1/\epsilon))}$.
\end{theorem}
\begin{proof}
Each subtable has the capacity to hold $(1+\epsilon)Nc/B$ items, and will
receive an expected $Nc/B$ items to store.  Let $X$ be the number of items
in the first subtable.  A standard Chernoff bound 
(e.g., \cite{mu-pcrap-05}[Theorem 4.4])  gives that 
$X$ is at most $(1+\epsilon/3)Nc/B$ with probability
bounded by
$$\Pr\left ( X \geq \left ( 1 + \frac{\epsilon}{3} \right) \frac{Nc}{B} \right) \leq e^{-Nc\epsilon^2/(27B)}.$$
With $B=O(n^\delta)$ for $\delta < 1$, we see that 
all subtables have at most $(1+\epsilon/3)Nc/B$ with probability subexponential in $N^{1-\delta}$.  By keeping
counters for each subtable, we can re-hash the items of {\em all} subtables in the rare case where a subtable
exceeds this number of items without affecting the expected insertion time by more than an $o(1)$ term.

The proof follows from Theorem 2 of \cite{dw-badtpcs-07}, by
noting that each subtable has space for at least
$(1+\epsilon/2)(1+\epsilon/3)Nc/B < (1+\epsilon)Nc/B$ items.  (In rare cases where an insertion fails,
we can re-insert all items {\em in a subtable} without affecting the expected insertion time by more than an $o(1)$ term.)
\end{proof}

It is likely this result could be improved
(see the remarks in \cite{dw-badtpcs-07}), but it is sufficient for
our purposes of showing that there is an insertion method for
external-memory cuckoo tables that uses a constant expected number of I/Os.

As noted in \cite{dw-badtpcs-07}, a more practical approach
is to use {\em random walk cuckoo hashing} in place of breadth first
search cuckoo hashing.  (For example, random walk cuckoo hashing is
used in all experiments in \cite{dw-badtpcs-07}.)  With
random walk cuckoo hashing, when an item cannot be placed, it kicks
out a single item in the bucket chosen uniformly at random.  Random
walk cuckoo hashing avoids the potentially large rare memory overhead
required of breadth first search, allowing instead a nearly stateless
solution.

More specifically, suppose a bucket $X$ is full when placing an item
$x$.  To reallocate items, we perform a random walk on the buckets,
starting from $X$, to find an augmenting path that has the net effect
of freeing up a location in $X$ (for $x$) while maintaining the
two-choice allocation rule for all the existing items in $C$.  Let
$Y$ denote the current node we are visiting in our random walk (which
is associated with a full bucket in the external-memory cuckoo
table---initially, the bucket $X$).  To identify the next node to
visit, we choose one of the items, $y$, in $Y$, uniformly at
random.  We then remove $y$ from $Y$ and insert the item $x$ waiting
to be inserted in $Y$.  We then let $y$ take over the role of $x$, and
attempt to place $x$ in the other bucket that is a possible
location for this item.  We repeat this process until we find a
non-full bucket or reach a pre-defined stopping condition.

For loads arbitrarily close to one, it is not known if there is a random walk cuckoo hashing scheme using
%$(1+\epsilon)N$ total space for $N$ items using
two bucket choices and
multiple items per bucket that similarly achieves expected constant
insertion time and logarithmic insertion time with high probability.
(This is given as an open question in \cite{dw-badtpcs-07}.)
Sadly, we do not resolve this question here.

However, for loads up to about $2/3$ we can utilize results by Panigrahy \cite{p-ehltma,p-hss-06} to
obtain such a random
walk cuckoo hashing scheme.  In Theorem~{2.3.2} of \cite{p-hss-06}, he
shows that for hash tables for $t$ items and load factors of $s$
satisfying $(2s)(1-e^{-2s}) < 1$, when the bucket size is 2, random
walk cuckoo hashing will succeed in inserting an item with a path
of length $O(\log t)$ with probability $1-O(1/t^2)$; his argument also
shows that this process has expected constant insertion time.  This
allows loads up to (approximately) $2/3$ using our partitioning technique
above.  In practice, we might expect this load to be improved
significantly in various ways.  First, we might ignore the partitioning,
and instead perform the random walk directly on the buckets with load
$B$.  Analyzing this process is difficult, in part
because of the greatly increased possibility of cycles in the cuckoo
graph.  Alternatively, we could perform the partitioning but allow the
random walk to stop early if there is room in the block $B$, rather
than the bucket for the corresponding subtable, effectively
multiplexing the bucket over subtable instantiations. We consider
these multiple variations in the simulations of Section \ref{sec:experiments}.

Finally, we point out that, as in a standard cuckoo hash table,
item look ups and removals use a worst-case constant number of I/Os.

\section{External-Memory Multimaps}
\label{sec:multimap}
In this section, we describe an extension of the external-memory cuckoo
hash table (as described in Section~\ref{sec:map})
that can be used to maintain a multimap in external
memory, so as to support fast dynamic access of a massive
data set of key-value pairs where some keys may have many associated
values.

\subsection{The Primary Structure}
To implement the multimap ADT, we begin with a primary structure that
is an external-memory cuckoo hash table storing just the set of keys.
In particular,
each record, $R(k)$, in $\cal T$, is associated with a specific key, $k$, and
holds the following fields:
\begin{itemize}
\item
the key, $k$, itself
\item
the number, $n_k$, of key-value pairs in $C$ with key equal to $k$ % \textbf{I don't think we ever use this -- JT;  But it seems odd NOT to include it, though we could get around it;  maybe some something}
%\item
%the number, $n_k$, of key-value pairs in $C$ with key equal to $k$ % \textbf{I don't think we ever use this -- JT}
\item
a pointer, $p_k$, to a block $X$ 
in a secondary table, ${\cal S}$,
that stores items in $C$ with key equal to $k$. 
Let $n_k$ denote the number of key-value pairs in $C$ with key equal to $k$. If $n_k< B$, then $X$
stores all the items with key equal to $k$ (plus possibly some items with
keys not equal to $k$).
Otherwise, if $n_k\ge B$, then $p_k$ points to a \emph{first} block
of items with key equal to $k$, with the other blocks of such items
being stored elsewhere in $\cal S$.
\end{itemize}

This secondary storage is an external-memory data structure we
are calling a \emph{multiqueue}.

\subsection{An External-Memory Location-Aware Multiqueue}
\label{sec:multiqueue}
\subsubsection{Overview}
The secondary storage
that we need in our construction is a
way to maintain a set $\cal Q$ of queues in external memory.
We assume
the \emph{header} pointers for these queues are 
stored in an array, $\cal T$,
which in our external-memory multimap construction is the
external-memory cuckoo hash table described above.

For any queue, $Q$, we wish to support the following operations:
\begin{itemize}
\item
enqueue($x,H$):
add the element $x$ to $Q$, given a pointer to its header, $H$. 
\item
remove($x$): remove $x$ from $Q$.
We assume in this case that each $x$ is unique.
\item
isMember($x$): determine whether $x$ is in some queue, $Q$.
\end{itemize}

In addition, we wish to maintain all these queues in a space-efficient
manner, so that the total storage is proportional to their total size.
To enable this, we store
all the blocks used for queue elements in a secondary table,
$\cal S$, of blocks of size $B$ each.
Thus, each header record, $H$ in $\cal T$, points to a block in $\cal S$.

Our intent is to store each queue $Q$ as a doubly-linked
list of blocks from $\cal S$.  Unfortunately, some queues in $\cal Q$ are too
small to deserve an entire block in $S$ dedicated to storing their elements.
So small queues must share their first block of storage with other small
queues until they are large enough to deserve an entire block of storage
dedicated to their elements.
Initially, all queues are assumed to be empty; hence, we initially
mark each queue as
being \emph{light}.
In addition, the blocks in $\cal S$ are initially empty; hence, we link the
blocks of $\cal S$ in a consecutive fashion as a doubly-linked list and
identify this list as being the \emph{free list}, $F$, for $\cal S$.

We set a heavy-size threshold at $B/3$ elements. When a queue $Q$ stored in a block $X$ reaches
this size, we allocate a block from $\cal S$ (taking a block off the free
list $F$) 
exclusively to store elements of $Q$ and we mark $Q$ as \emph{heavy}. 
Likewise, to avoid wasting space as
elements are removed from a queue, we require any heavy queue $Q$ to
have at least $B/4$ elements. If a heavy queue's size falls below this
lower threshold, then we mark $Q$ as being light again
and we force $Q$ to go back to sharing it space with other small queues.
This may in turn involve returning a block to the free list $F$.
In this way,
each block $X$ in $\cal S$ will either be empty or will have all its
elements belonging to a single heavy queue or as many as $O(B)$ light queues.
In addition, these rules also imply that 
$O(B)$ 
element insertions are required to take a queue from the light state to the 
heavy state
and
$O(B)$ 
element removals are required to take a queue from the heavy state to the light
state.

If a block $X$ in $\cal S$ is being used for light queues, 
then we order the elements in $X$ according to their respective queues.
%and we include a pointer back to the header for each such queue.
Each block for a heavy queue $Q$ stores previous and next pointers to the
neighboring blocks in the linked list of blocks for $Q$, with 
the first such block pointing back to the header record for $Q$.
As we show, this organization allows us to maintain our size and label
invariants during execution of enqueue and remove operations.

One additional challenge is that we want to support the remove($x$) operation
to have a constant I/O complexity.
Thus, we cannot afford to search through a list of blocks of a queue looking
for an element $x$ we wish to remove.
So, in addition to the table $\cal S$ and its free list, $F$,
and the headers for each queue in $\cal Q$,
we also maintain an external-memory 
cuckoo hash table, $\cal D$, to be a dictionary that
maps each queue element $x$ to the block in $\cal S$ that stores $x$.
This allows our multiqueue to be \emph{location-aware}, that is, to
support fast searches to locate the block in $\cal S$ that is holding
any element $x$ that belongs to some queue, $Q$.%, in our multiqueue.

We will call any block in $\cal S$ containing fewer than $B/4$ items \emph{deficient}. 
In order to ensure that our multiqueue uses total storage proportional to its total size,%uses $O(N/B)$ blocks of memory (i.e. constant load),
we will enforce the following two rules. Together, these rules guarantee that there are $O(N/B)$ deficient blocks in $\cal S$, and hence our 
multiqueue uses $O(N/B)$ blocks of memory.
\begin{enumerate}
\item Each
block $Y$ in $\cal T$ stores a pointer $d$, called the deficient pointer, to 
a block $d(Y)$; the identity of this block is allowed to vary over time. We ensure that at all times, $d(Y)$ is the only (possibly) deficient block associated with $Y$ that stores light queues.
\item Each heavy queue $Q$ also stores in its header block a deficient pointer $d$ to a block $d(Q)$. At all times, $d(Q)$ is the only (possibly) deficient block devoted to
storing values for $Q$.
\end{enumerate} 

\subsubsection{Full Description}
\label{sec:basic}
For the remainder of this subsection, we describe how to implement all multiqueue operations
to obtain constant  \emph{amortized} expected or worst-case runtime. We show how to
deamortize these operations in Section \ref{sec:app}.
% {\bf But I don't get, here, if we have to search through the actual block in order
% to now do the removal?  If so, is it then $O(B)$ time to do a removal?  Is that
% what we wanted?  Similarly, I'm confused how then we do an insert.  Do we ``compact''
% each block after a removal so we can insert new items at the end of a block?  That
% seems reasonable but then also everything takes $O(B)$ time -- I'm assuming we're not
% representing elements in a block with a doubly-linked list, but perhaps we should?
% Also, is it then $O(B)$ time to look up an element as well?
% 
% I think I had somehow originally thought that a block consisting of ``light queues'' would
% itself be a cuckoo hash table of items -- this would be constant-factor wasteful in space but
% would immediately make look ups/deletes constant time, with failure probabilities depending on
% the size of B.  I can get that might not be desirable.}

\paragraph{The Split Action.}
As we perform enqueue operations, a block $X$ may overflow 
its size bound, $B$.
In this case, we need to split $X$ in two, which we do 
by allocating a new block $X'$ from $\cal S$ (using its free list). 
We call $X$ the \emph{source} of the split, and $X'$ the \emph{sink} of the split.
We then proceed depending on whether $X$ contains 
elements from light queues or a single heavy queue.
\begin{enumerate}
\item
$X$ contains elements from light queues.
We greedily copy elements from $X$ into $X'$ until $X'$ has size 
has size at least $B/3$, keeping the elements from the same light queue together.
Note that each light queue has less than $B/3$ elements, so this split will result in at
least a $1/3$--$2/3$ balance. 

Of course, to maintain our invariants, we must change the header records 
from $X$ to $X'$ for any queues that we just moved to $X'$.
We can achieve this by performing a look-up in $\cal T$ for each key corresponding to a queue that was moved
from $X$ to $X'$, and modifying its header record, which requires $O(B)$ I/Os. %Notice that
%most of these pointers should be located in the same block in $\cal T$ that contains $H$, 
%which is already cached (the only pointers for which this is not true are ones that have been
%moved by the cuckoo hashing scheme implementing the primary structure $\cal T$). Thus, in practice this will
%not require many I/Os.
Similarly, in order to support location awareness,
we must also update the dictionary $\cal D$. So, for each element $x$
that is moved to $X'$, we look up $x$ in $\cal D$ and update its pointer to now
point to $X'$.
In total this costs $O(B)$ I/Os. \eat{ but we can charge these,
in an amortized sense, to the $O(B)$
enqueue operations that caused 
$X$ to grow from size at most $2B/3$ to at least $B$.
(These enqueue operations will never be charged in this way again.)}
\item
$X$ contains elements from a single heavy queue $Q$. In this case, we move \emph{no} elements,
and simply take a block $X'$ from the free list and insert it as the head of the list of blocks devoted to $Q$, changing
the header record $H$ in $\cal T$ to point to $X'$. 
We also change the deficient pointer $d$ for $Q$ to point to $X'$, and insert into $X'$ the element that caused the split.
This takes $O(1)$ I/O operations in total.

%We copy half of the elements from $X$ to $X'$.
%We keep the previous pointer (which may be the header)
%for $Q$ pointing to $X$, but we add pointers so that $X'$
%is now between $X$ and $X$'s successor.
%As in Case~1, we also update pointers in $D$ for each moved element $x$ to 
%now point to $X'$. 
%This, of course, costs us $O(B)$ I/Os, but we can charge these,
%in an amortized sense, to the $O(B)$
%enqueue operations that caused 
%$X$ to grow from size at most $B/2$ to at least $B$.
%(These enqueue operations will never be charged in this way again.)
\end{enumerate}

So, to sum up, when a block holding light queues 
results from a split (source or sink), it has size
at least $B/3$ and at most $2B/3$. When a block holding elements
from a heavy queue $Q$ is split, no items are moved and a block is taken
from the free list and inserted as the new header block of the heavy queue; the new header then contains only one item,
and is identified by the deficient pointer of $Q$. 

\paragraph{The Enqueue Operation.}
Given the above components,
let us describe how we 
perform the enqueue and remove operations.
We begin with the enqueue($x,H$) operation.
We consider how this operation acts, depending on a few cases.

\begin{enumerate}
\item
The queue for $H$ is empty (hence, $H$ is a null pointer 
and its queue is light).
In this case, we examine the block $Y$ from $\cal T$ to which $H$ belongs.
If $d(Y)$ is null, we first take a block $X$ of the free list and set $d(Y)$ to $X$ before continuing. 
We follow the deficient pointer for $Y$ to a block $X'$, and add $x$ to $X'$. If this causes
the size of $X'$ to reach $B$, then we split $X'$ as described above. 
\item
The queue $Q$ for $H$ is not empty. We proceed according to two cases.
\begin{enumerate}
\item If $Q$ is a light queue, we follow $H$ to its block $X$ in $\cal S$ and add $x$ to $X$. If this
brings the size of $Q$ above $B/3$, we perform a \emph{light-to-heavy transition}, taking a block $X'$ off the free list, moving all elements in $Q$
to $X'$, and marking $Q$ as heavy. If this brings the size of $X$ below $B/4$, we process $X$ as in the remove operation below.
\item If $Q$ is a heavy queue, we add $x$ to $X=d(Q)$, the (possibly) deficient block for $Q$.
If this brings the size of $X$ to $B$, then we split $X$, as described
above.
\end{enumerate}
\end{enumerate}
Once the element $x$ is added to a block $X$ in $\cal S$, we then add
$x$ to the dictionary $\cal D$, and have its record point to $X$.

\paragraph{The Remove and isMember Operations.}
%Let us next consider the remove($x$) and isMember($x$) operations.
In both of these operations, we look up $x$ in $\cal D$ 
to find the block $X$ in $\cal S$ that
contains $x$.
In the isMember($x$) case, we complete the operation
by simply looking for $x$ in $X$.
In the
remove($x$) operation, we do this look up and then remove $x$ from $X$
if we find $x$.
If this causes $Q$ to become empty, then we update its header, $H$,
to be null.
In addition, if this operation causes the size of $X$ to go below $B/4$, then
we need to do some additional work, based on the following cases:
\begin{enumerate}
\item
$Q$ is a heavy queue.
\begin{enumerate}
\item
If $X$ is the only block for $Q$, then $Q$ should now be considered a light
queue; hence, we continue processing it according to the case listed
below where $X$ contains only light queues. We refer to the entirety of this action
as a \emph{heavy-to-light} queue transition.
\item
Otherwise, if $X=d(Q)$, then we are done because $d(Q)$ is allowed to be deficient. 
If $X \neq d(Q)$,
we proceed based on the following two cases:
\begin{enumerate}
\item \emph{$d$-alteration action:} If the size of $d(Q)$ is at least $2B/3$, we simply update $Q$'s deficient pointer, $d$, to point to $X$ instead of $d(Q)$.
\item \emph{Merge action}:
If the size of $d(Q)$ is less than $2B/3$, then we move all 
of the elements of $X$ into $d(Q)$ and we update 
the pointer in $\cal D$ for each moved element. $X$ is returned to the free list.
We call $X$ the source of the merge, and $d(Q)$ the sink.
(Note that in this case, the size of $d(Q)$ becomes at most $11B/12$.)
\end{enumerate}
\end{enumerate}
\item
$X$ contains light queues (hence, no heavy queue elements).  
In this case, we visit the header $H$ for $Q$. 
Let $Y$ denote the block containing $H$. 
\begin{enumerate} 
\item If $X=d(Y)$ we are done,
since $d(Y)$ is allowed to be deficient.
\item
If $X \neq d(Y)$, let $Z$ be the size of $d(Y)$. 
\begin{enumerate}
\item \emph{$d$-alteration action:} If $Z \geq 2B/3$ then we simply update $d$ to point to $X$ instead of $d(Y)$.
\item \emph{Merge action}:
If $Z < 2B/3$, then we merge the elements in $X$ into $d(Y)$, 
which now has size at most $11B/12$, and update pointers in
$\cal D$ and $\cal T$ for the elements that are moved. We return $X$ to the free list. We call 
$X$ the source of the merge and $d(Y)$ the sink.
\end{enumerate}
\end{enumerate}
\end{enumerate}
If a block $X'$ is pointed to by any deficient pointer $d$, it is helpful to think of this 
as ``protection" for $X'$ from being the source of a merge. Once $X'$ is afforded this protection,
it will not lose it until its size is at least $2B/3$ (see the $d$-alteration action).
At a high level, this will allow us to argue that if $X$ and $X'$ are respectively 
the source and sink of a merge action,
neither $X$ nor $X'$ will be the source of a subsequent merge or split operation until it is the target of $\Omega(B)$ enqueue
or remove operations, even though $X'$ may have size very close to the deficiency threshold $B/4$. 

\subsubsection{Amortized I/O complexity}
\label{sec:amort}
\eat{In total this costs $O(B)$ I/Os, but we can charge these,
in an amortized sense, to the $O(B)$
enqueue operations that caused 
$X$ to grow from size at most $2B/3$ to at least $B$.
(These enqueue operations will never be charged in this way again.)}

We now argue formally that enqueue$(x, H)$ and remove($x$) take $O(1)$ amortized time. Notice that the only actions that result in the movement of items between blocks are light-to-heavy and heavy-to-light queue transitions, merge actions, and split actions for blocks containing light queues. Notice for splits involving heavy queues, we perform $O(1)$ I/O operations 
in the worst case,
and do not need to perform an amortized analysis.  %Also notice heavy-to-light queue transitions are handled by simply marking the queue as light and then treating the remainder of the update as a merge action, so we need not consider their I/O complexity separately).  

We first argue that light-to-heavy queue transitions as well as heavy-to-light transitions contribute $O(1)$ amortized I/Os to enqueue operations. Indeed, a light-to-heavy queue transition
requires $O(B)$ I/Os in total: we require $O(B)$ I/Os to move $O(B)$ items from $X$ to $X'$ and update pointers in $\cal D$ and $\cal T$, and $O(B)$ additional I/Os to process $X$ as in a remove operation if this causes the size of $X$ to fall below $B/4$. Each such heavy-to-light transition
must be preceded by at least $B/12$ enqueue operations to bring the queue from size at most $B/4$ to size at least $B/3$, so we can charge these $O(B)$ I/Os to these enqueue operations. These enqueue operations will never be charged again. Similarly, a heavy-to-light queue transition requires $O(B)$ I/Os, which we can charge to the (at least) $B/12$ removals that caused $Q$'s size
to fall from $B/3$ to $B/4$; these removals will never be charged again.

Since we have accounted for the I/Os caused by light-to-heavy and heavy-to-light queue transitions, we may ignore all I/Os caused by
these transitions through the remainder of the argument.
We now argue that merge and split actions contribute $O(1)$ amortized I/Os as well, beginning with merge actions.

Suppose $X$ and $X'$ are respectively 
the source and sink of a merge action. We claim that neither $X$ nor $X'$ will be the source of a subsequent merge or split operation until it is the target of $\Omega(B)$ enqueue
or remove operations.
Indeed, notice that after doing a merge action as a part of our
processing of a remove operation,
the sink will
contain at most $11B/12$ elements and will be equal to $d(Y)$  or $d(Q)$, and the source is on the free list.
As $d(Y)$ and $d(Q)$ are protected from merges, it would take at least $B/12$ enqueues or removals in these
blocks before they would be sources of another split or merge operation.

Likewise, after performing a split of a block containing light queues as a part of an enqueue operation,
both source and sink will be of size at least $B/3$ and at most
$2B/3$.
Thus, it would take at least $B/12$ enqueues or removals in these
blocks before they would be sources of another split or merge operation.

Therefore, in an amortized analysis, we can charge the $O(B)$ I/Os
performed in a split or merge action to the previous $O(B)$ operations
that caused one of these blocks to shrink to size $B/4$ or grow to
size $B$. These enqueues and removals will never be charged again.

The arguments of the last two paragraphs are depicted graphically in Figure \ref{flow}. Assuming no light-to-heavy or heavy-to-light transitions
take place (we may assume this because we have separately accounted for the I/O cost of these transitions), we depict a subgraph of
the state diagram for any block $X$. Specifically, we depict all state transitions caused by any action that results in the movement of items from one block to another; for brevity, we omit the effects of any actions that do not result in the movement of items. We refer to any state corresponding to a source
of a merge or split action as a ``source state." It is clear that in the subgraph depicted in Figure \ref{flow}, there is no directed path from any non-source state to any
source state. Given this fact, it is a straightforward exercise to confirm that the only paths from non-source states to source states in the full state diagram (assuming no light-to-heavy or heavy-to-light transitions)
include at least $B/12$ enqueue or remove operations to $X$.
\eat{ in the state diagram
from any non-source state to any source state. As all non-source states in Figure
that for any block $X$ that is a merge sink, merge source, split sink, or split source, at least $B/12$ enqueue or remove operations are necessary}

\begin{figure*}
\begin{center}
\includegraphics[width=4in]{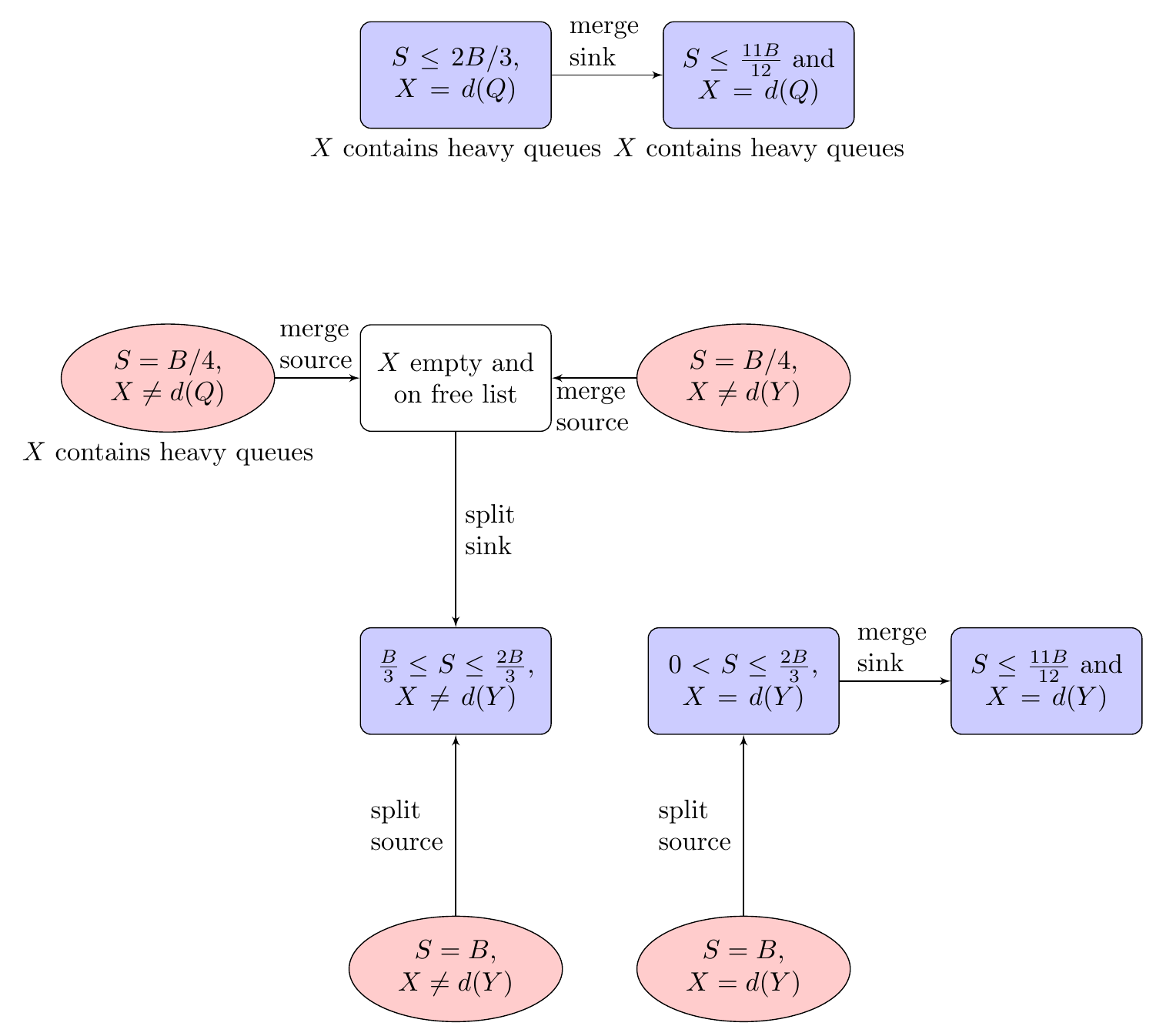}
\caption{A subgraph of the state diagram for any block $X$, depicting all state transitions caused by merge or split actions. $S$ denotes the size of $X$. Ovals denote source states,
while rectangles denote non-source states. Unless otherwise noted, any state depicted is for a block containing light queues.}
\label{flow}
\end{center}
\end{figure*}

\subsection{Deamortizing Multiqueue Operations}
\label{sec:app}
 %Likewise, as soon as a heavy queue $Q$ fell to size $B/4$, we moved it from its block $X$ to a block $X'=d(Y)$,
%for some deficient pointer $d$; after this move, $X'$ might have size very close to $B$, or might even require immediate splitting.
We now explain how to deamortize the multiqueue operations of the previous section. First, notice that the only actions that result in the movement of items between blocks are merge actions, split actions for blocks devoted to heavy queues, light-to-heavy queue transitions, and heavy-to-light queue transitions. We will require the follow property: for any action resulting in the movement of items from source block $X$ to sink block $X'$, neither $X$ nor $X'$ will
be the source of any subsequent action requiring the movement of items until it is the target of at least $B/12$ enqueue or remove operations.

First, we describe some modifications to the light-to-heavy and heavy-to-light queue transitions that are necessary to ensure this property is satisfied. We begin with light-to-heavy transitions. Previously, as soon as a light
queue $Q$ grew to size $B/3$, it was moved from its block $X$ to a block $X'$ devoted exclusively to $Q$; this could cause the size of $X$ to fall close to or below $B/4$, and $X$ could therefore be the source of a merge shortly after (or immediately upon) the light-to-heavy transition. Because this clearly does not satisfy the required property, we will do away with an explicit light-to-heavy transition action, and instead fold this functionality into the split action as follows. 

We leave unmodified the split action for blocks $X$ devoted to heavy queues, as well as for blocks $X$ containing only light queues in which none of the queues have size greater than $B/3$. It is easy to see in both of these cases that the required property is satisfied, as in the first case (split for blocks devoted to heavy queues) no items are moved, and in the second case both the source and sink of the split have size between $B/3$ and $2B/3$. %, and hence neither will be the source of a split or merge action for until it is the target of at least $B/12$ enqueue or remove operations by arguments identical to
%Section \ref{sec:amort}.

However, if the source $X$ of the split move contains a queue $Q$ of size at least $B/3$, we proceed according to the following cases.
\begin{enumerate}
\item $X$ contains a queue $Q$ of size between $B/3$ and $2B/3$. We take a new block $X'$ off the free list and move all items in $Q$
to $X'$, marking $Q$ as heavy and updating the affected pointers in $\cal T$ and $\cal D$. After this split action, both $X$ and $X'$ have size between $B/3$ and $2B/3$, and hence neither will be the source of a split action or merge action until it is the target of at least $B/12$ enqueue or remove operations. 
\item $X$ contains a queue $Q$ of size greater than $2B/3$. Let $\cal I$ denote the items in $X$ that are not in $Q$. We proceed according to the following cases.
	\begin{enumerate} \item If $d(Y)$ has size less than $B/3$, we leave $Q$ in $X$ and mark it as heavy. In addition, we transfer all items in $\cal I$ to $d(Y)$, and update all affected pointers in $\cal T$ and $\cal D$. After the split, $X$ is devoted to $Q$ and has size at least $2B/3$. $X'$ now has size at most $2B/3$, and moreover $X'=d(Y)$ and thus $X'$ is protected from being the source of a merge. It therefore requires at least $B/12$ inserts or removals to $X$ or $X'$ before either can be the source of any action requiring the movement of items between blocks.
	\item  If $d(Y)$ has size greater than $B/3$, we leave $Q$ in $X$ and mark it as heavy. We take a new block $X'$ off the free list and transfer all items in $\cal I$ to $X'$. We update all affected pointers in $\cal T$ and $\cal D$, and modify the deficient pointer $d$ of $Y$ to point to $X'$. The source block $X$ is devoted to $Q$ and has size at least $2B/3$. $X'$ has size $|{\cal I}|\leq B/3$, and moreover $X'=d(Y)$ and thus $X'$ is protected from being the source of a merge. It therefore requires at least $B/12$ inserts or removals to $X$ or $X'$ before either can be the source of any subsequent action requiring the movement of items between blocks.  %since $d(Y)$ is never the source of a merge operation, and $X'$ will equal $d(Y)$ until its size exceeds $B/3$, the second required property is satisfied as well.
	\end{enumerate}
\end{enumerate}

Let us now explain a small modification we must make to the heavy-to-light transitions in order to satisfy the required property. 
Observe that it is possible for a queue $Q$ to undergo a heavy-to-light transition shortly after the final two blocks $X$ and $X'$ devoted to $Q$ are merged into one. For example, it is possible that $X'=d(Q)$ contains one item before the merge and $B/4+1$ items after the merge; if one item is subsequently removed from $X'$, $Q$ will undergo a heavy-to-light transition, and our required property will not be satisfied. This is the only setting in which a deficient pointer fails to ``protect" a block from being merged. To circumvent this difficulty, we modify the heavy-to-light queue transition to only occur when the size of the heavy queue falls below $B/6$ rather than $B/4$. With this in hand, the arguments of Section \ref{sec:amort} suffice to show that any merge action or heavy-to-light transition satisfies our required property. This completes the description of all modifications necessary to ensure the required property is satisfied by all actions.

We now explain how to deamortize the operations of Section \ref{sec:multiqueue}, which all required $O(1)$ amortized time. 
The only actions requiring $\omega(1)$ I/O operations in Section \ref{sec:multiqueue} were split actions, merge actions, heavy-to-light transitions, and
light-to-heavy transitions that caused elements from a source block $X$ to be moved to a sink block $X' \neq X$ (the latter have now been replaced with a modified split operation). These actions
required $O(B)$ I/O operations to immediately update all affected pointers in $\cal T$ and $\cal D$.
To deamortize these operations, we immediately move the elements from $X$ to $X'$, but do not immediately update 
any pointers in $\cal T$ and $\cal D$. Instead, we create a pointer $p(X)$ from $X$ to $X'$,
allowing us to spread out
the updates to $\cal D$ and $\cal T$ over many operations as follows. 

We will ensure that any block $X$
need point to at most one block $X'$ at any time; specifically, any time a split action or merge action
causes items to move from block $X$ to block $X'$, we will overwrite the old value of $p(X)$ with the new value. To clarify, when a block $X$ is sent to the free list as a result of 
a merge operation, it must maintain its pointer $p(X)$ throughout its time on the free list; it is only safe to overwrite $p(X)$ when items are 
once again moved
from $X$ to another block $X'$.

We will also ensure that no queue is ever moved more than once before its header in $\cal T$ and the records for all of its
key-value pairs in $\cal D$ are brought up-to-date. Given this fact, 
%enqueue($k,v)$, remove$(k, v)$, and isMember$(k, v)$ require $O(1)$ I/O operations:
if we ever follow a pointer from $\cal T$ or $\cal D$ to a block $X$, and the corresponding
item is not in $X$, we need only look in $p(X)$ for the item as well. 

To this end, we associate with each block $X'$ in $\cal S$ a bit-array of length $O(B)$ indicating which items in
$X'$ have up-to-date pointers in $\cal T$ and $\cal D$. Any time items are moved into $X'$ as a result
of a split or merge action, we set the corresponding bits in the bit-array of $X'$ to 0, indicating
these items are not up-to-date. Further, we modify the
enqueue($k,v)$ and remove$(k, v)$ operations such that if $(k, v)$ is stored in block $X$, then we 
update the pointers in $\cal T$ and $\cal D$
of up to 12 items in $X$ and 12 items from
from $p(X)$ that are not up-to-date. We then mark these items as up-to-date. This requires only $O(1)$ I/O operations for each enqueue$(k,v)$ or remove($k,v)$
function call.

We finally argue that each time items from a block
$X$ to be moved to a block $X'$,
it is safe to overwrite $p(X)$ with a pointer to $X'$. Indeed, we carefully argued above that all actions resulting in a movement of items from source block $X$ to sink block $X'$ satisfy our required property. It is easy to see that this implies neither $X$ nor $X'$ will be the source of another sink or merge until it is the target of at least $B/12$ enqueue or remove operations. By that point, all items in $X$ (or $X'$) and $p(X)$ (or $p(X')$) will be up-to-date, so it safe to overwrite $p(X)$ (or $p(X')$).

%Furthermore, notice that $X'$ will only be the sink of a merge if $X'$ has size at most $2B/3$, and this
%is the only way $X'$ can receive more than one element in a single operation. Thus, 
%once $X'$ has size $Z > 2B/3$, it will not receive more than one item per operation until it experiences
%at least $B-Z$ insertions, or $Z-Z/4$ removals. Consequently, if $X'$ is the sink of a split or the sink of a merge,
%it will not be the \emph{source} of  another sink or merge until at least $B/12$ items have been inserted or deleted from $X'$.

We obtain the following theorem.

\begin{theorem} \label{thm:multiqueue}
We can implement a location-aware multiqueue so that 
the remove($x$) and isMember($x$)
operations each use $O(1)$ I/Os,
and the enqueue($x,H$) operation uses $O(1+t(N))$ expected I/Os,
where $t(N)$ is the expected number of I/Os needed to perform an
insertion in an external-memory cuckoo table of size $N$.
\end{theorem}

It should be clear from our description that, except for trivial cases
(such as having only a constant number of elements), the space
requirements of our multiqueue implementation is within a constant
factor of the optimal.  We have not attempted to optimize this factor, though
there is some flexibility in the multiqueue operations (such as when to do a split)
that would allow some optimization.
We study these tradeoffs in Section \ref{sec:experiments}.

% the space tradeoffs, both in theory and in practice,
%appears an interesting subject for future work.

\section{Combining Cuckoo Hashing and Location-Aware Multiqueues}
\label{sec:combine}
In this section, we describe how to construct an efficient
external-memory multimap implementation by combining the 
data structures described above.
The result is 
a cuckoo hash table in external memory so as to support constant
expected-I/O insertions and optimal findAll and removeAll
operations.

\begin{figure*}[!thb]
\begin{center}
\includegraphics[width=4in, trim=0.75in 2.5in 0.75in 0.5in, clip]{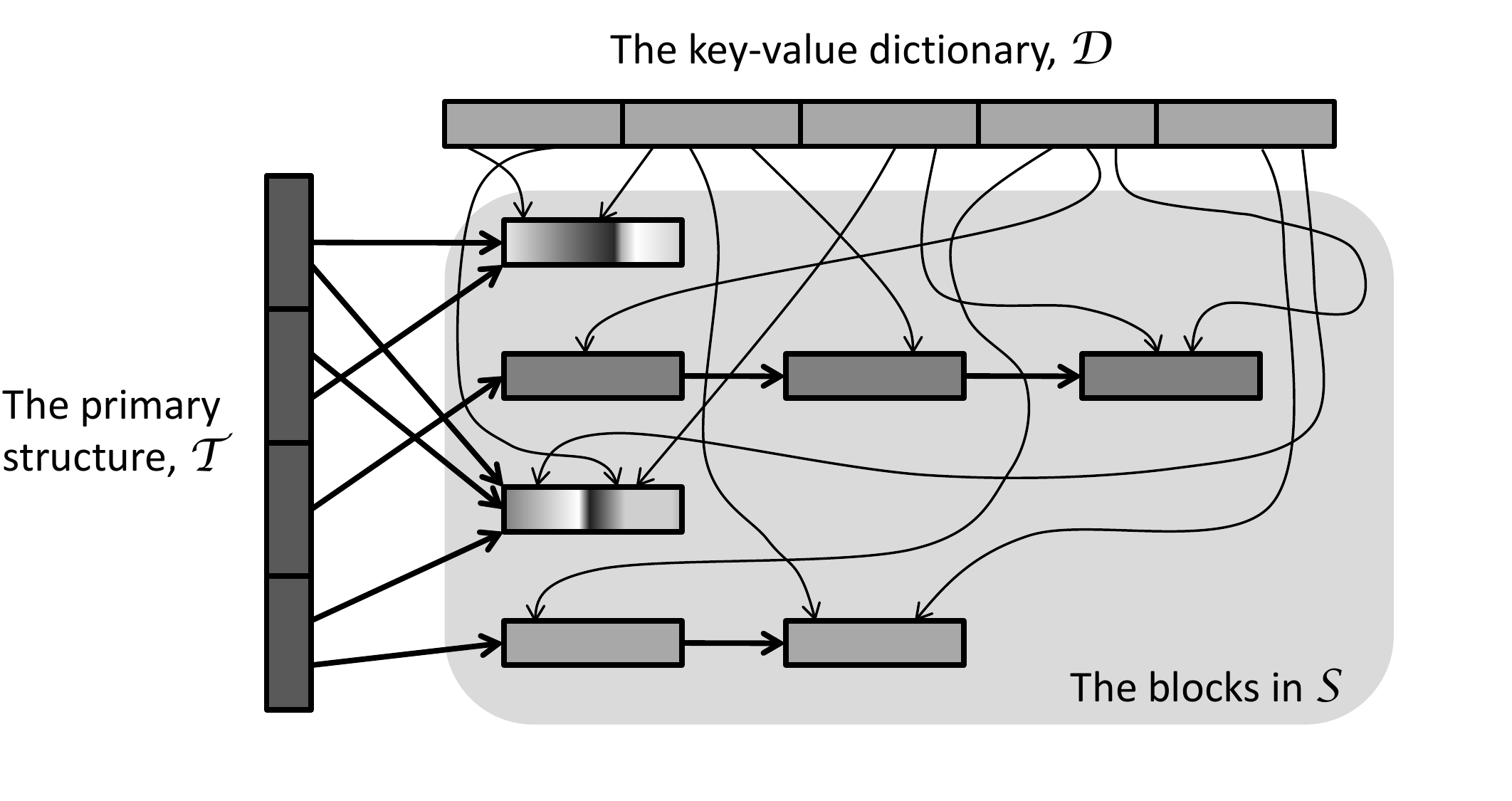}
\caption{The external-memory multimap, online version.}
\label{fig:buckets}
\end{center}
\end{figure*}

We store an external-memory cuckoo hash table, as
described above, as our primary structure, $\cal T$,
with each record pointing to a block
in a multiqueue, $\cal S$, having an auxiliary dictionary, $D$,
implemented as yet another external-memory cuckoo hash table.
We then perform each
of the operations of the multimap ADT as follows.

\begin{itemize}
\item
insert$(k,v)$: To insert the key-value pair, $(k,v)$,
we first perform a look up for $k$ in $\cal T$.
If there is already a record for $k$ in $\cal T$, we increment its count. We
 then follow its pointer to the appropriate block $X$ in $\cal S$, (in the deamortized implementation,
the queue for $k$ may reside in $p(X)$ rather than $X$),
and add the pair $(k,v)$ to $\cal S$, as in the enqueue multiqueue method.
Otherwise we insert $k$ into $\cal T$ with a null header record and count 1 and then add the pair $(k, v)$ 
to $\cal S$ as in the enqueue multiqueue method.
\item
isMember$(k,v)$: This is identical to the isMember$(k,v)$ multiqueue operation. \eat{We do a look up in $D$ to determine whether the key-value pair
$(k,v)$ exists as a record in $\cal S$, and return 
a Boolean value encoding the success or failure of this
look up.} 
\item
remove$(k,v)$: To remove the key-value pair, $(k,v)$, from $C$,
%we first perform a look up for $k$ in $\cal T$.
%If there is no record for $k$ in $\cal T$, then we return an error
%condition.
%Otherwise,
%if there is a record for $k$ in $\cal T$, then we 
we perform a look up for $(k,v)$ in $\cal D$. If there is no record for $(k,v)$ in $\cal D$, we
return an error condition. Otherwise, we follow this pointer
to the appropriate block $X$ of $\cal S$ holding the pair $(k,v)$ (in the deamortized implementation,
if $(k,v)$ is not in $X$, we may have to look in $p(X)$ as well).
%If we find no such pair, then we return an error condition.
%Otherwise, 
We remove the pair $(k,v)$ from $\cal S$ and $\cal D$ as in the remove multiqueue method, and decrement its count.
\item
findAll$(k)$: To return the set of all key-value pairs
in $C$ having key equal to $k$, we perform a look up for $k$ in $\cal T$,
and follow its pointer to the appropriate block of $\cal S$ (in the deamortized implementation,
the queue for $k$ may reside in $p(X)$ rather than $X$). If this
is a light queue, then we just return the items with key equal to
$k$.
Otherwise, we return the entire block and all the other blocks of
this queue as well.
\item
removeAll$(k)$: We give here a constant amortized time implementation, and explain
in Section \ref{sec:removeall} how to deamortize this operation. To remove from $C$ all key-value pairs
having key equal to $k$, we perform a look up for $k$ in $\cal T$,
and follow its pointer to the appropriate block $X$ of $\cal S$ (in the deamortized implementation,
the queue for $k$ may reside in $p(X)$ rather than $X$). If this
is a light queue, then we remove from $X$ all items with key equal to
$k$ and remove all affected pointers from $\cal D$; if this causes $X$ to become deficient, we 
perform a merge action or $d$-alteration action as in the remove multiqueue method. If
this is a heavy queue, we walk through all blocks of this queue and remove all items from these blocks and
return each block to the free list. We also remove all affected pointers from $\cal D$. 
Finally, we remove the header record for $k$ from $\cal T$, which implicitly sets the count of $k$ to zero as well. We charge, in an amortized sense,
the work for all the I/Os to the insertions that added these
key-value pairs to $C$ in the first place.
\item count($k$): Return $n_k$, which we track explicitly for all keys $k$ in $\cal T$.
\end{itemize}

\subsection{Deamortizing removeAll($k$)}
\label{sec:removeall}
The removeAll($k$) operation of Section \ref{sec:combine} required $O(1)$ amortized I/O operations in the worst case without altering the capacity of our structure.
We now describe a deamortized implementation that also requires $O(1)$ I/O operations and does not alter the capacity. We perform a look up for $k$ in $\cal T$. If no record 
is found, we are done. Otherwise we follow its pointer to the header of its queue $Q$. We remove all items in $Q$ from $\cal S$, and set $k$'s pointer in $\cal T$ to null. This completes
the operation; notice we do \emph{not}
update any records in $\cal D$ at this time. Instead, we explain the modifications necessary to handle the existence of ``spurious" pointers in $\cal D$ (i.e. pointers for $(k, v)$ pairs which were deleted in a removeAll operation) with an $O(1)$ increase in the I/O cost of the insert$(k,v)$, remove($k, v$), isMember($k, v$), and findAll($k)$ operations.

First, we describe a function isSpurious($k,v$) that requires $O(1)$ I/O operations and determines whether an entry $(k,v)$ in $\cal D$ is spurious. 
isSpurious($k,v$) first peforms a look up in $\cal D$ for $(k, v)$. If no record for $(k, v)$ exists, we return false. Otherwise, we follow the pointer for $(k, v)$ to a block $X$ in $\cal S$ and search $X$ for $(k, v)$. If a record is found we return false.
Otherwise, we follow the pointer $p(X)$ (described in Section \ref{sec:app}) to a block $X'$ and search for $(k, v)$ in $X'$.  If it is found,
we return false, otherwise we return true. 

We now describe how to modify the insertion method of our external-memory cuckoo hash table $\cal D$ so that the presence of spurious pointers does not decrease 
the table's capacity. First, when inserting a key-value pair $(k,v)$ into $D$, we begin by doing a look up in $\cal D$ for $(k, v)$. If a record for $(k, v)$ exists, we call isSpurious$(k,v)$. If this function returns false, we return an error condition. Otherwise, we remove the record for $(k,v)$ from $\cal D$ before proceeding.   This ensures that at all times there is only one entry for each pair $(k,v)$ in $\cal D$.

Second, we modify the BFS-based insertion procedure of Theorem \ref{thm:cuckoo} as follows. For each bucket visited by the BFS, we 
call isSpurious$(k,v)$ for all pairs $(k,v)$ residing in the bucket. If this function returns true for any pair $(k,v)$, we delete $(k,v)$ from $\cal D$ and
insert the new pair in its place. This ensures that no spurious entry in $\cal D$ ever prevents another entry from
being inserted, i.e., the spurious entries will have no effect on the capacity of the table. Since the buckets in the cuckoo hashing algorithm of Theorem \ref{thm:cuckoo} have constant size, calling isSpurious$(k,v)$ on a
bucket requires just $O(1)$ I/O operations.

With this in hand, we finally describe how to modify the insert$(k,v)$, remove($k, v$), isMember($k, v$), and findAll($k)$ operations
to handle the presence of spurious entries in $\cal D$ with only an $O(1)$ increase in the I/O complexity of each operation.

\begin{enumerate}
\item
insert$(k,v)$: Works unmodified.
\item isMember($k,v$): We call the previous implementation of isMember($k,v)$ as well as the function checkSpurious($k,v)$. We return true if and only if the former
returns true and the latter returns false.
\item remove$(k,v)$:  We 
perform a look up for $(k,v)$ in $\cal D$. If none is found, we return an error condition. Otherwise, we call the function isSpurious($k,v)$. If this returns true,
we return an error condition. Otherwise, we 
call the old implementation of remove$(k,v)$. 
\item findAll$(k)$: Works unmodified.
\end{enumerate}

We finally obtain the following theorem.

\begin{theorem}
One can implement the multimap ADT in external memory using $O(N/B)$ blocks of memory with 
I/O performance as shown in Table~\ref{tbl:bounds}.
\end{theorem}

 \section{Experimental Results}
\label{sec:experiments}
We performed extensive simulations of our algorithm in order to explore how various settings of the design parameters affect I/O complexity and space usage, for both our basic algorithm (Section \ref{sec:multiqueue}) and our deamortized algorithm (Section \ref{sec:app}).

We simulated a cache of size $M=512$ KB with blocks of size 4 KB. 
Our simulated cache used the least-recently used page replacement rule. 
When reporting the number of I/Os, we count only transfers from disk to cache; 
each such transfer is preceded by a transfer from cache to disk of the least recently used cache page, 
and we do not count this transfer in our reported values. We drew keys from a universe of size $2^{20} \approx 1$ million, 
using $4$ bytes to store each key, and 8 bytes to store each value. We did not explicitly store queues as 
doubly-linked lists, but instead laid them out as arrays within their blocks, with a marker representing 
the end of one queue and the beginning of another; this allowed us to avoid storing expensive pointers for 
these lists. We used 4 bytes to represent all pointer values in $\cal D$ and $\cal T$. We did not 
charge for storing the counts associated with each key because we do not need to store these counts explicitly
except to achieve $O(1)$ I/O operations for the count$(k)$ operation (and moreover we can achieve 
this by only storing explicit counts for heavy queues, as the count of a light queue $Q$ can be 
obtained in O(1) I/Os by finding $Q$'s unique block in $\cal S$ via a lookup in $\cal T$ and then counting how many items $Q$ contains). 

All results presented use random-walk cuckoo hashing with two hash
functions and buckets capable of storing 4 KBs of data; we found that
using the partitioning technique of Theorem \ref{thm:cuckoo} to
implement cuckoo hashing required slightly more space (and I/O
complexity was comparable) because the hash tables had slightly
smaller capacity. 
For our hash tables $\cal D$ and $\cal T$, we allotted a space overhead of 
$\epsilon=0.07$; we found this was even more overhead than strictly necessary.
We also ran a full set of experiments using three
hash functions to implement cuckoo hashing, but found that two hash
functions was sufficient due to the large bucket size; we found using two hash
functions instead of three saved about 1 I/O per insert and remove
operation.
%any realistic frequency distributions are skewed, with a few items having high frequency and many items having low frequency. 
To capture realistic frequency distributions, which are often skewed, we generated all keys for insertions from a Zipfian distribution; in a Zipfian distribution with parameter $\alpha$, the frequency of the $k$'th most frequent item is proportional to $k^{-\alpha}$. The larger $\alpha$, the more skewed the frequency distribution.

Our goal was to identify the steady-state behavior of our data structure. In all experiments, we performed a sequence of 1 million insertions, followed by a sequence of 8 million alternating insert($k,v)$ and remove$(k,v)$ operations.  For each remove operation, the pair for removal was selected uniformly at random from the table. 

\subsection{Basic Implementation}
\label{exp:am}

Space usage results from the basic algorithm are shown in Figures \ref{fig-hash2a} and \ref{fig-hash2b}.  The vertical line represents the point at which we completed $2^{20} \approx 1$ million insertions and began alternating insertions and deletions. $\alpha$ denotes the Zipfian parameter, $B/\beta$ denotes the light-to-heavy queue transition threshold, and $B/\gamma$ is the deficiency parameter (i.e. blocks of size less than $B/\gamma $ are declared deficient). Notice we experiment with more aggressive settings of $\beta$ and $\gamma$ for $\alpha=1.1$.  

Across all parameter settings, we achieved steady-state loads of between $.33$ and $.39$, where we defined the load to be $S/(12 \times 2^{20})$, where $S$ is the number of bytes used by pages not on the free list in our algorithm, and $12 \times 2^{20}$ is the minimum number of bytes required to explicitly store all $2^{20}$ key-value pairs in the structure. Notice with $4$ KB blocks, $12 \times 2^{20}$ bytes corresponds to just over 3,000 blocks of memory. 

%\vspace{-0.2cm}
%\newlength{\figwidth}
%\setlength{\figwidth}{0.46\textwidth}
%\begin{figure*}
%\centering
%\subfigure[Space usage for $\alpha=.99$.]
%{
%\includegraphics[width=\figwidth]{block_usage_alpha=099.pdf}
%\label{fig1}
%}
%\subfigure[center][Space usage for $\alpha=1.1$.]
%{
%\includegraphics[width=\figwidth]{block_usage_alpha=11.pdf}
%\label{fig2}
%}
%\vspace{-0.4cm}
%\caption{Results from simulations of an implementation of our basic multimap algorithms, using random-walk cuckoo hashing with three hash functions.}
%\end{figure*}
%\vspace{-0.2cm}

\newlength{\figwidth}
\setlength{\figwidth}{0.46\textwidth}
\begin{figure*}[!h]
\centering
\subfigure[Space usage for $\alpha=.99$.]
{
\includegraphics[width=\figwidth]{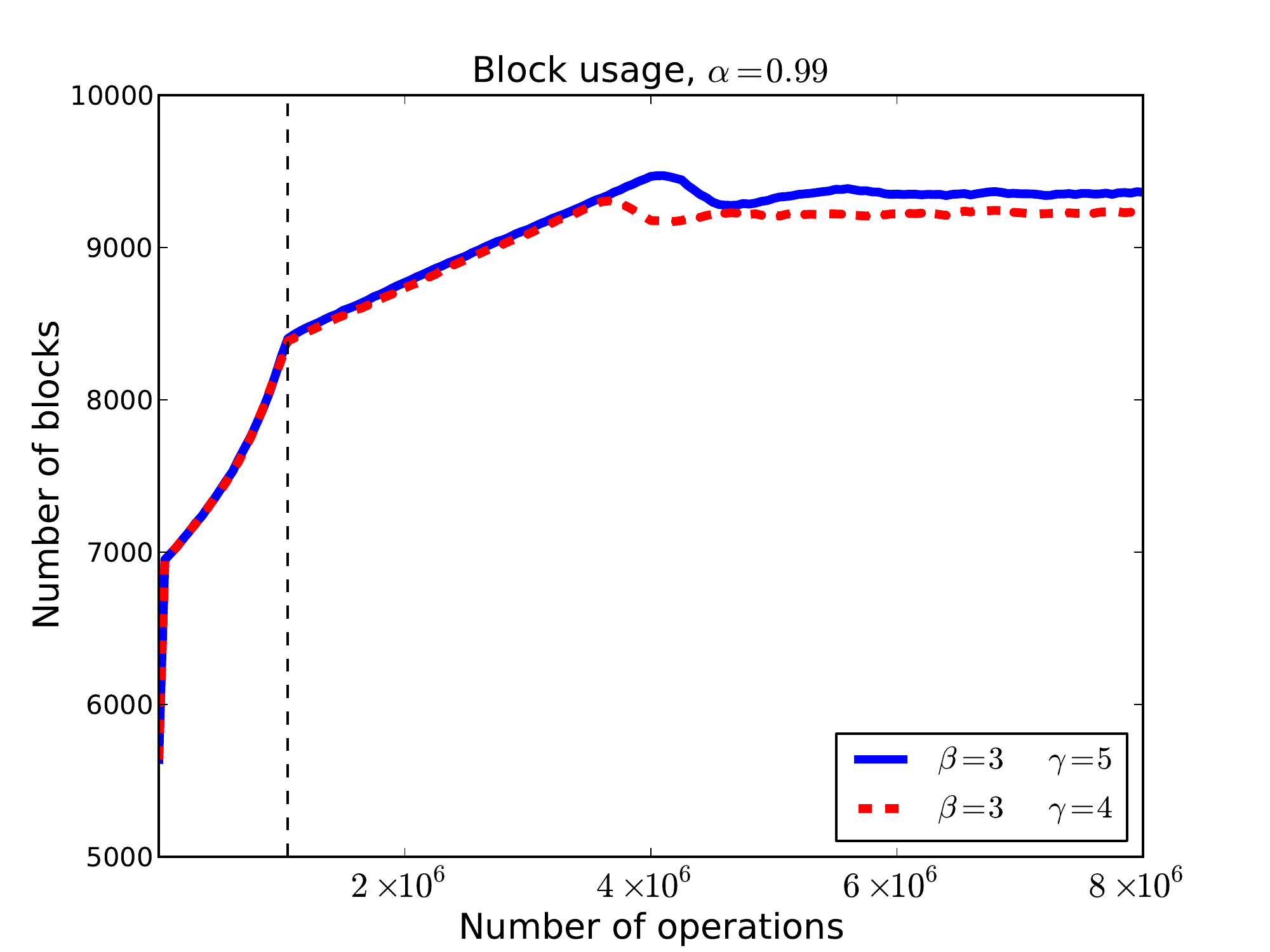}
\label{fig-hash2a}
}
\subfigure[center][Space usage for $\alpha=1.1$.]
{
\includegraphics[width=\figwidth]{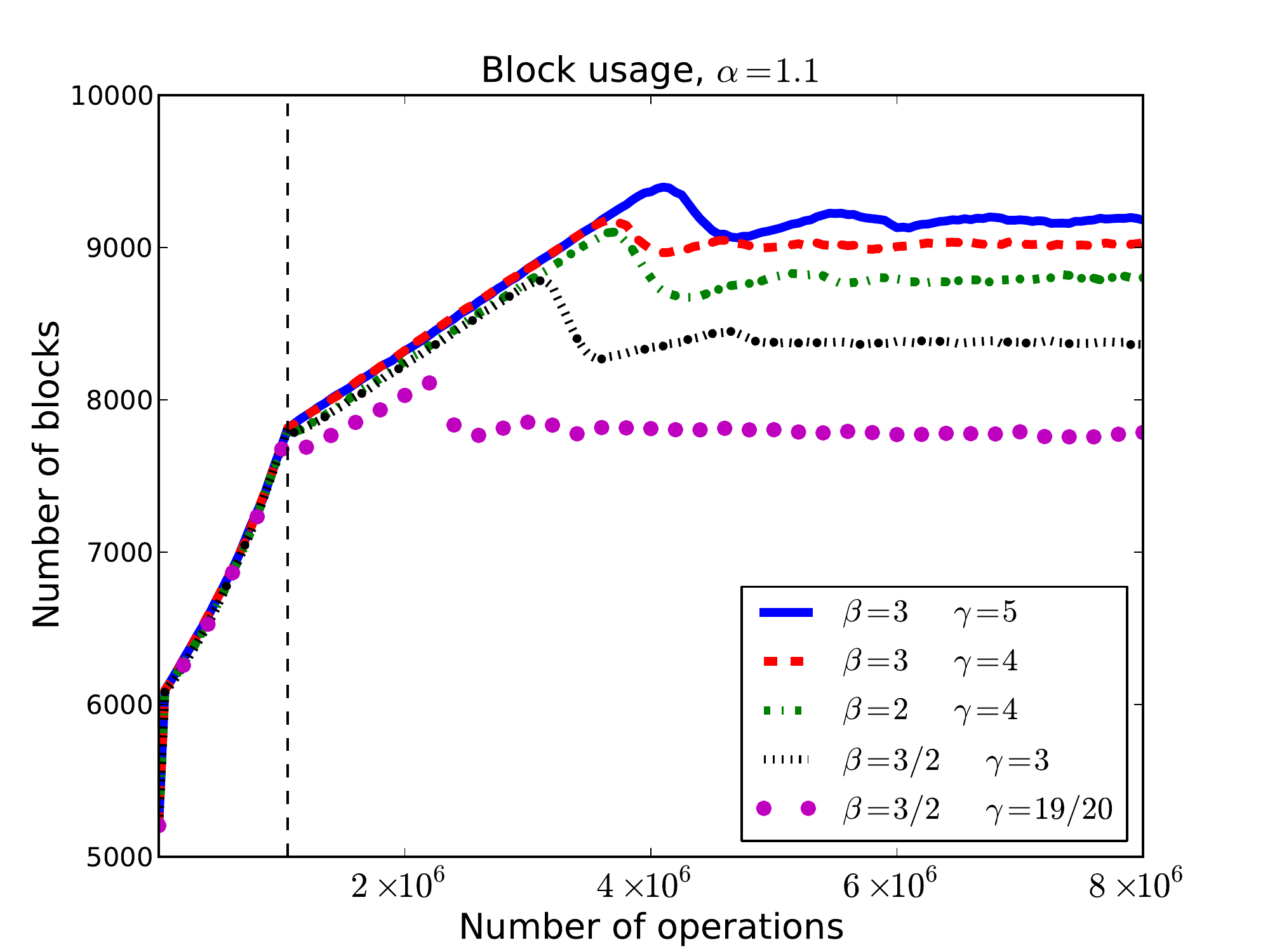}
\label{fig-hash2b}
}
%\vspace{-0.2cm}
\caption{Results from simulations of an implementation of our basic (amortized) multimap algorithms.}
\end{figure*}
%\vspace{-0.2cm}

I/O statistics from representative settings of parameters are in Table~\ref{tab}. 
%Complete tables are in Appendix \ref{app:exp}. 
A smaller $\gamma$ results in improved space usage as a smaller $\gamma$ implies that we perform merge actions more aggressively. Similarly, a smaller $\beta$ implies we are more reluctant to tie up entire blocks devoted to a single heavy queue, and thus yields improved space usage. %By choosing $\beta$ and $\gamma$ smaller than we show in Figures \ref{fig1}, we expect to achieve higher loads.

In the basic algorithm, the average cost over all insert and remove operations is extremely low: about $3.5$ I/Os per operation. However, as depicted in Table~\ref{tab} the cost distribution is bimodal -- the vast majority (over 99.9\%, except for $\gamma$ very close to 2) of operations require about 4 I/Os, but a small fraction of operations require several hundred. The maximum number of I/Os ranges between $400$ and $650$.

These high-cost operations are due to split and merge actions. 
The deamortized implementation displays substantially different behavior, with no operation requiring more than a few dozen I/Os (see Section \ref{exp:deam}).  Notice we tested parameter values for which the theoretical bounds on I/O complexity do not hold; for example, with $\gamma=19/10$, a merge may immediately follow a split.

\vspace{-0.25cm}
\begin{table}[!h]
\small
\centering
%\scalebox{0.97}{
\subtable[Overall]{
\begin{tabular}{|c|c|c|c|c|c|} 
\hline
$\alpha$ & $\beta$ & $\gamma$ & Mean & Std Dev & Max \\
 & & & I/Os & I/Os & I/Os \\
\hline
0.99 & 3 & 5 & 3.53 & 4.24 & 639 \\ 
0.99 & 3 & 4 & 3.52 & 4.59 & 625 \\ 
\hline
1.10 & 3 & 5 & 3.17 & 4.29 & 398 \\ 
1.10 & 3 & 4 & 3.23 & 4.90 & 401 \\ 
1.10 & 2 & 4 & 3.20 & 5.27 & 403 \\ 
1.10 & 3/2 & 3 & 3.25 & 6.73 & 534 \\ 
1.10 & 3/2 & 19/10 & 3.68 & 14.81 & 536 \\ 
\hline
\end{tabular}
}

\subtable[$\le$ 15 I/Os]{
\begin{tabular}{|c|c|c|c|c|c|}
\hline
$\alpha$ & $\beta$ & $\gamma$ & \% of & Mean & Std Dev \\
 &  &  & Ops &  I/Os &  I/Os \\
\hline
0.99 & 3 & 5 & 99.96 & 3.46 & 0.97 \\ 
0.99 & 3 & 4 & 99.96 & 3.44 & 0.96 \\ 
\hline
1.10 & 3 & 5 & 99.95 & 3.08 & 1.13 \\ 
1.10 & 3 & 4 & 99.94 & 3.12 & 1.14 \\ 
1.10 & 2 & 4 & 99.95 & 3.09 & 1.14 \\ 
1.10 & 3/2 & 3 & 99.95 & 3.10 & 1.14 \\ 
1.10 & 3/2 & 19/10 & 99.83 & 3.09 & 1.13 \\ 
\hline
\end{tabular}
}
\subtable[ $>$ 15 I/Os]{
\begin{tabular}{|c|c|c|c|c|c|}
\hline
$\alpha$ & $\beta$ & $\gamma$ & \% of & Mean & Std Dev \\
 &  &  & Ops &  I/Os &  I/Os \\
\hline
0.99 & 3 & 5 & 0.04 & 203.59 & 84.64 \\ 
0.99 & 3 & 4 & 0.04 & 199.69 & 82.29 \\ 
\hline
1.10 & 3 & 5 & 0.05 & 181.66 & 67.12 \\ 
1.10 & 3 & 4 & 0.06 & 183.51 & 68.81 \\ 
1.10 & 2 & 4 & 0.05 & 224.45 & 74.44 \\ 
1.10 & 3/2 & 3 & 0.05 & 279.58 & 72.17 \\ 
1.10 & 3/2 & 19/10 & 0.17 & 354.52 & 79.05 \\ 
\hline
\end{tabular}
}
%}
\caption{\label{tab}I/O statistics for our basic (amortized) implementation (a) overall, (b) for operations requiring up to 15 I/Os, and (c) for operations requiring more than 15 I/Os.}
\end{table}

\subsection{Deamortized Implementation}
\label{exp:deam}

Figures \ref{fig-d-hash2a} and \ref{fig-d-hash2b} presents space usage results for the 
deamortized implementation, following the same protocol as the amortized 
experiments (Section \ref{exp:am}).  We achieved loads of about $.33$ to $.35$ 
for basic parameter values ($\gamma=4$ and $\gamma=5$). We also experimented with very
high settings of $\gamma$, where we trade-off increased space usage for improved I/O complexity.

\setlength{\figwidth}{0.46\textwidth}
\begin{figure*}[!h]
\centering
\subfigure[Space usage for $\alpha=.99$.]
{
\includegraphics[width=\figwidth]{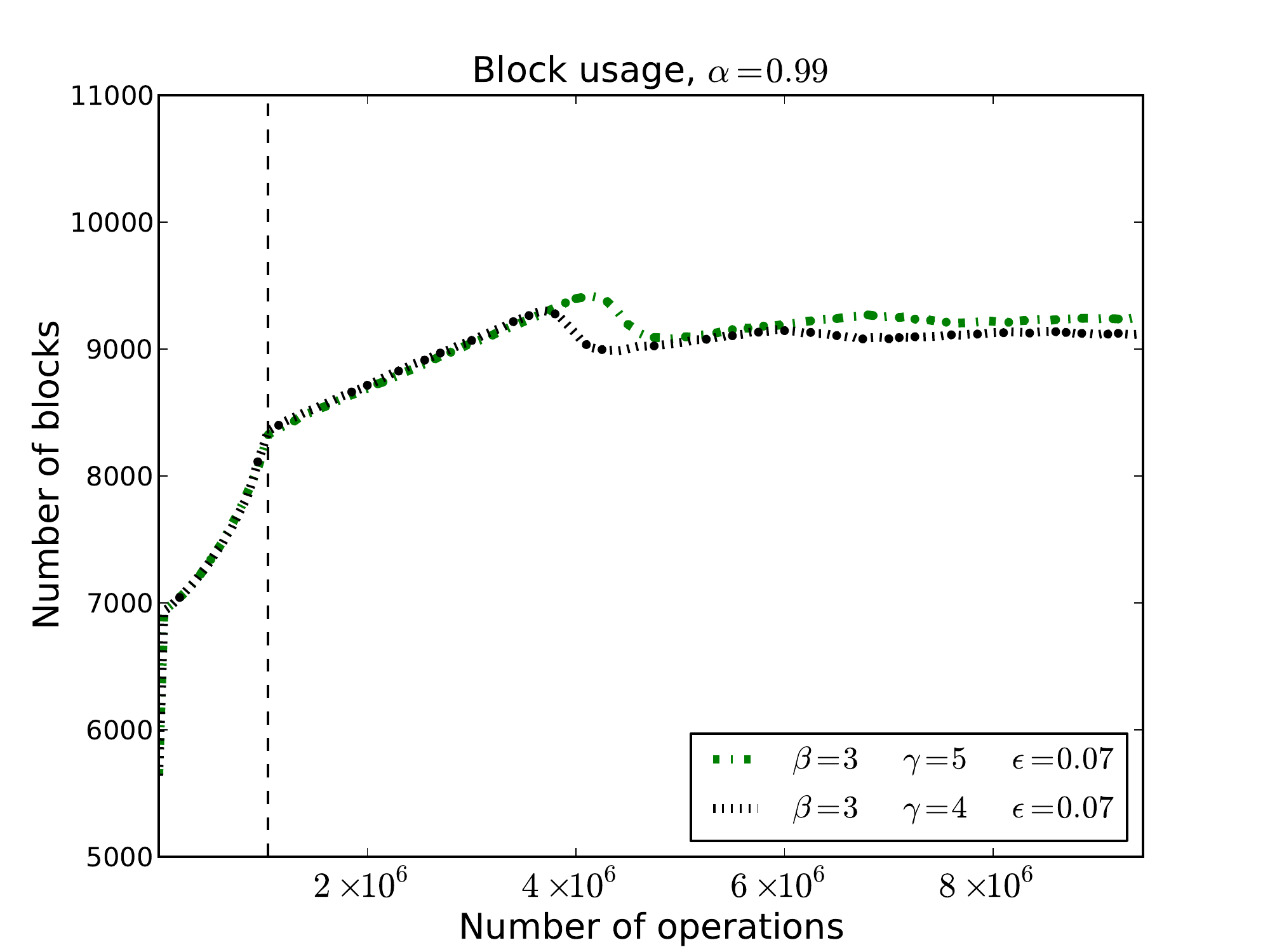}
\label{fig-d-hash2a}
}
\subfigure[center][Space usage for $\alpha=1.1$.]
{
\includegraphics[width=\figwidth]{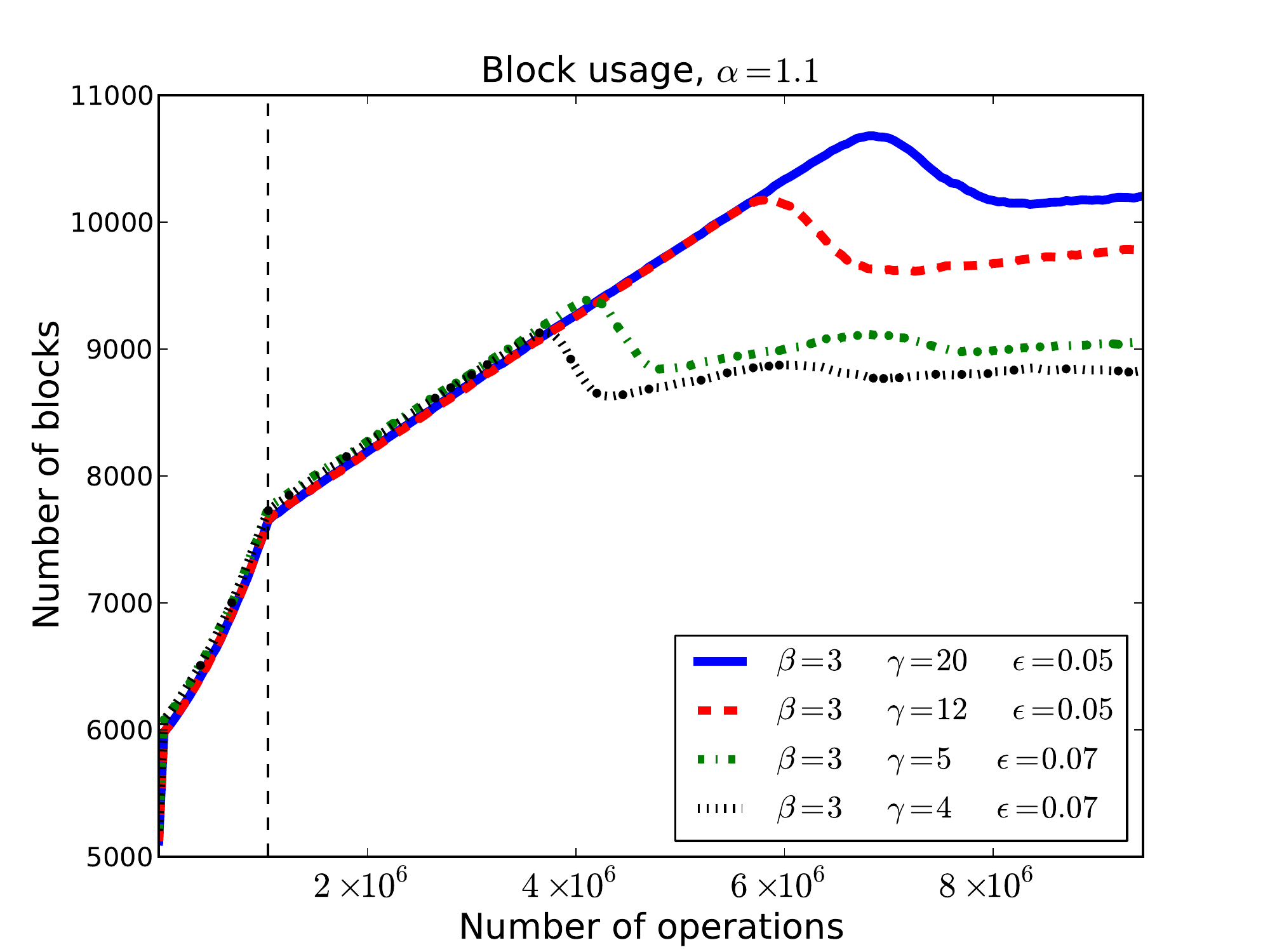}
\label{fig-d-hash2b}
}
\vspace{-0.2cm}
\caption{Results from simulations of an implementation of our deamortized multimap algorithms.}
\end{figure*}
%\vspace{-0.2cm}

More interesting is the I/O complexity of the deamortized implementation,
shown in Table~\ref{tab2}. 
We see that in stark contrast to the bimodal cost distribution of the basic 
implementation, the deamortized implementation never requires more than a few 
dozen I/Os for any given operation. 
Moreover, even the average I/O complexity of the deamortized implementation is 
significantly better than that of the basic implementation, with an improvement
of at least $0.5$ I/Os per operation, for parameters where we have a direct comparison. 
We attribute much of this improvement to the modified split rule, 
which makes light-to-heavy queue transitions significantly less expensive.
Note that the maximum number of I/Os for any operation in our deamortized experiments
across all parameters, is at most 43 -- an order of magnitude below 
the maximum for our basic algorithm. 

%We note briefly that our implementation of the remove operation does a lookup in 
%$\cal T$ in order to update a count.  If we instead do not maintain these 
%counts, we found that this reduces the mean number of I/Os per remove operation
%by about $0.5$.

In Table~\ref{tab2}, we also display the breakdown in I/O complexity between 
inserts and remove operations. We see that removes are about twice as expensive
as inserts; this is not unexpected.  An insert requires a look up in $\cal T$, 
followed by loading the header page for the queue $Q$, 
an insert into $\cal D$, and then possibly a split.  
Due to the skewness of our input data, these first two steps can be free, 
as these pages are often already in the cache. 
A remove requires a look up in $\cal D$, followed by
loading the appropriate page in $\cal S$, and then possibly a merge.  
In contrast to inserts, the first two steps are rarely free. 

\begin{table}[!h]
\centering
\small
%\scalebox{0.97}{
\subtable[Overall]{
\begin{tabular}{|c|c|c|c|c|c|} 
\hline
$\alpha$ & $\beta$ & $\gamma$ & Mean & Std Dev & Max\\
 & & & I/Os & I/Os & I/Os \\
\hline
0.99 & 3 & 5 & 2.96 & 1.75 & 42 \\ 
0.99 & 3 & 4 & 2.99 & 1.83 & 43 \\ 
\hline
1.10 & 3 & 20 & 2.60 & 1.66 & 41 \\ 
1.10 & 3 & 12 & 2.59 & 1.71 & 42 \\ 
1.10 & 3 & 5 & 2.66 & 1.96 & 42 \\ 
1.10 & 3 & 4 & 2.66 & 2.06 & 43 \\ 
\hline
\end{tabular}
}

\subtable[$\le$ 15 I/Os]{
\begin{tabular}{|c|c|c|c|c|c|}
\hline
$\alpha$ & $\beta$ & $\gamma$ & \% of & Mean & Std Dev \\
 &  &  & Ops &  I/Os &  I/Os \\
\hline
0.99 & 3 & 5 & 99.81 & 2.90 & 1.23 \\ 
0.99 & 3 & 4 & 99.78 & 2.92 & 1.24 \\ 
\hline
1.10 & 3 & 20 & 99.90 & 2.58 & 1.40 \\ 
1.10 & 3 & 12 & 99.88 & 2.56 & 1.39 \\ 
1.10 & 3 & 5 & 99.78 & 2.60 & 1.41 \\ 
1.10 & 3 & 4 & 99.73 & 2.58 & 1.41 \\ 
\hline
\end{tabular}
}
\subtable[$>$ 15 I/Os]{
\begin{tabular}{|c|c|c|c|c|c|}
\hline
$\alpha$ & $\beta$ & $\gamma$ & \% of & Mean & Std Dev \\
 &  &  & Ops &  I/Os &  I/Os \\
\hline
0.99 & 3 & 5 & 0.19 & 31.53 & 3.16 \\ 
0.99 & 3 & 4 & 0.22 & 31.44 & 3.23 \\ 
\hline
1.10 & 3 & 20 & 0.10 & 31.53 & 3.43 \\ 
1.10 & 3 & 12 & 0.12 & 31.38 & 3.26 \\ 
1.10 & 3 & 5 & 0.22 & 31.23 & 3.02 \\ 
1.10 & 3 & 4 & 0.27 & 31.12 & 3.05 \\ 
\hline
\end{tabular}
}

\subtable[Insert operations]{
\begin{tabular}{|c|c|c|c|c|c|}
\hline
$\alpha$ & $\beta$ & $\gamma$ & Mean & Std Dev & Max \\
 &  &  &  I/Os &  I/Os & \\
\hline
0.99 & 3 & 5 & 2.28 & 1.88 & 40 \\ 
0.99 & 3 & 4 & 2.32 & 2.00 & 42 \\ 
\hline
1.10 & 3 & 20 & 1.86 & 1.68 & 41 \\ 
1.10 & 3 & 12 & 1.85 & 1.76 & 42 \\ 
1.10 & 3 & 5 & 1.94 & 2.16 & 42 \\ 
1.10 & 3 & 4 & 1.94 & 2.32 & 41 \\ 
\hline
\end{tabular}
}
\subtable[Remove operations]{
\begin{tabular}{|c|c|c|c|c|c|}
\hline
$\alpha$ & $\beta$ & $\gamma$ & Mean & Std Dev & Max \\
 &  &  &  I/Os &  I/Os & \\
\hline
0.99 & 3 & 5 & 3.80 & 1.09 & 42 \\ 
0.99 & 3 & 4 & 3.82 & 1.12 & 43 \\ 
\hline
1.10 & 3 & 20 & 3.53 & 1.08 & 40 \\ 
1.10 & 3 & 12 & 3.52 & 1.09 & 42 \\ 
1.10 & 3 & 5 & 3.57 & 1.15 & 42 \\ 
1.10 & 3 & 4 & 3.55 & 1.18 & 43 \\ 
\hline
\end{tabular}
}
%}
\caption{\label{tab2}I/O statistics for the deamortized implementation, for (a) overall, 
(b) for operations requiring up to 15 I/Os, (c) for operations requiring more than 15 I/Os,
(d) for insert operations, and (e) for remove operations.}
\end{table}

We note that we have tested our performance against a working
commerical database product, which places key-value pairs in a hash
table, but moves key-value pairs associated with a key to a B-tree
when the number of values associated with a key becomes large.
Preliminary tests suggest this approach yields a slightly smaller
average number of memory accesses per operation, but in turn runs with
significantly more memory.  We expect that both implementations could
be optimized significantly, and each approach might be preferable in
different circumstances.

\section{Conclusion}
We have described an efficient external-memory
implementation of the multimap ADT, which generalizes the inverted
file data structure that is useful for supporting search engines.
Our methods are based on new expected-time bounds for performing 
updates in block-based cuckoo hash tables as well as an
external-memory multiqueue data structure. In addition to proving 
theoretical bounds on the I/O complexity of our implementation,
we demonstrated experimentally that our data structure is able to trade off 
constant factors in space against the time to
perform operations in well-understood ways.  

One direction for future work is to consider efficient in-memory
algorithms for multimaps, an area that seems to not have been given
significant attention.  Another natural direction 
would be to derive improved high-probability bounds for block-based
cuckoo hash tables.  In particular, improved analysis of random walk
cuckoo hashing in this setting is worthwhile.  These are natural
extensions of open problems in the theory of cuckoo hashing.

\eat{
It is also worth noting that using cuckoo hash tables we can, in
theory, trade off constant factors in space against the time to
perform operations in fairly well-understood ways.  Our multiqueue
description is not currently optimized for space.  It would be
worthwhile to attempt to understand the tradeoffs for this
structure, in order to attempt to minimize the amount of unused space
(not residing in blocks in the free list) at any point.
}

\subsection*{Acknowledgments}
We thank Margo Seltzer for several helpful discussions.

{\raggedright
\bibliographystyle{abbrv}
\bibliography{goodrich,extra2,cuckoo,cuckoo2,range}
}

\end{document}